\documentclass[11pt]{article}
\usepackage{amsmath,amssymb,amsthm,mathtools}
\usepackage[margin=1in]{geometry}
\usepackage[numbers,sort&compress]{natbib}
\usepackage[colorlinks=true,linkcolor=blue,citecolor=blue,urlcolor=blue]{hyperref}
\usepackage{microtype}

\newtheorem{theorem}{Theorem}[section]
\newtheorem{lemma}[theorem]{Lemma}
\newtheorem{proposition}[theorem]{Proposition}
\newtheorem{corollary}[theorem]{Corollary}

\theoremstyle{definition}

\theoremstyle{remark}
\newtheorem{remark}[theorem]{Remark}

\newcommand{\cH}{\mathcal H}
\newcommand{\cF}{\mathcal F}
\newcommand{\Ralg}{\mathcal R}
\newcommand{\diam}{\operatorname{diam}}
\newcommand{\Var}{\operatorname{Var}}
\newcommand{\dom}{\operatorname{dom}}
\newcommand{\Tr}{\operatorname{Tr}}
\newcommand{\id}{\operatorname{id}}
\newcommand{\spn}{\operatorname{span}}

\title{Energy-Constrained Commutator Variance for Weyl Pairs}
\author{Hassan Nasreddine\thanks{\texttt{hassan.nasreddine@fortisec.net}}}
\date{}

\begin{document}
\maketitle

\begin{abstract}
Let $W$ be a representation of the Weyl relations on a separable Hilbert space and write
$\mathcal R_W(K)=\{W(f):f\in K\}''$.  If the symplectic pairing of $K_1$ and $K_2$ is nonzero,
there is a self-adjoint unitary $B\in\mathcal R_W(K_2)$ such that, for every normal state $\rho$,
one can choose a self-adjoint unitary $A_\rho\in\mathcal R_W(K_1)$ with
$\Var_\rho(-i[A_\rho,B])=4$.  Thus the optimized mean-input-energy-constrained coefficient is
$4$ at every admissible energy threshold.  For the fixed trigonometric witness $h$, the identity
for $4-h^2$ reduces the deficit to a phase-fixed problem for commuting squared Weyl translations.
For the one-mode quadratic energy
$G_M=\tfrac12R^TMR-\tfrac12\sqrt{\det M}$ with $M>0$ and $u^T\Omega v=\pi$, the resulting all-state
variational problem has the asymptotic law
\[
4-\gamma_{G_M,E}(h)
=\frac{1}{4E}\bigl(u^TMu+v^TMv+2\pi\sqrt{\det M}\bigr)+O(E^{-2}).
\]
The lower bound holds over all normal states satisfying the mean-energy constraint, and a localized
Zak construction attains the same coefficient.  For a general second-quantized energy
$G=d\Gamma(H_1)$, the fixed-witness deficit is $\Theta(E^{-1})$ whenever at least one witness
direction has a nonzero positive-energy component; if both directions are zero modes, saturation is
exact at every positive threshold.
\end{abstract}

\section{Introduction and overview}

Energy constraints are useful in infinite-dimensional quantum systems because they restrict the
states used to probe a system without truncating its operator algebra.  This distinction is
central in the energy-constrained norm framework of Shirokov and related work
\cite{Shirokov2018,BeckerDattaLamiRouze2021}.  A different, complementary condition concerns
the operation itself: an energy-limited channel or dynamics maps finite-energy inputs to outputs
whose energy can be controlled quantitatively \cite{vanLuijk2025EnergyLimited}.  The two
requirements need not agree.

The present paper studies this distinction for an order defect attached to two von Neumann
algebras.  For bounded self-adjoint $h$ and a positive self-adjoint reference observable $G$, put
\begin{equation}
\gamma_{G,E}(h)
:=2\sup_{\rho\in\mathfrak S_{G,E}}\sqrt{\Var_\rho(h)},
\qquad
\mathfrak S_{G,E}:=\{\rho\ge0:\Tr\rho=1,\ \Tr(G\rho)\le E\},
\label{eq:gamma-intro}
\end{equation}
where $E>\inf\sigma(G)$.  For von Neumann algebras $M,N\subset B(\cH)$ define
\begin{equation}
\Gamma_{G,E}(M,N)
:=\sup_{\substack{a=a^*\in M,\ \|a\|\le1\\ b=b^*\in N,\ \|b\|\le1}}
\gamma_{G,E}(-i[a,b]).
\label{eq:Gamma-intro}
\end{equation}
The unconstrained analogue is
\begin{equation}
\Gamma(M,N)
:=\sup_{\substack{a=a^*\in M,\ \|a\|\le1\\ b=b^*\in N,\ \|b\|\le1}}
\diam\sigma(-i[a,b]).
\label{eq:Gamma-unconstrained}
\end{equation}
These quantities arise as mixed second-order coefficients of channel-order defects; the
operator-algebraic and channel-theoretic formulation is developed in
\cite{Nasreddine2026DualityArxiv}.  That work proves
\[
0\le \Gamma_{G,E}(M,N)\le \Gamma(M,N)\le4,
\qquad
\Gamma_{G,E}(M,N)\uparrow\Gamma(M,N)
\quad(E\to\infty),
\]
see \cite[Thm.~4.14]{Nasreddine2026DualityArxiv}, and identifies the unrestricted value $4$ for a
nonzero symplectic pairing in the symmetric Fock setting
\cite[Prop.~5.4 and Prop.~5.7]{Nasreddine2026DualityArxiv}.  The questions considered here
are more specific: whether the endpoint is already attained at a fixed energy threshold, and what
energy law remains when the bounded witness is fixed in advance.  The first question has an exact
statewise answer.  The second leads to an all-state variational problem for commuting squared Weyl
translations.  In one mode with a positive quadratic reference energy, its leading coefficient can
be computed explicitly; for a general $d\Gamma(H_1)$ one obtains an inverse-energy/zero-mode
dichotomy.

Related extremal phenomena occur in the operator-algebraic theory of Bell correlations
\cite{SummersWerner1987Generic,Landau1987Bell}.  There the algebras commute and the extremal
quantity is a correlation functional.  Here the nonzero symplectic pairing makes the Weyl
algebras noncommuting, and the quantity of interest is the variance of their commutator.

\subsection{Main results}

\begin{theorem}[Statewise Weyl saturation]
\label{thm:intro-saturation}
Let $W$ be a representation of the Weyl relations of $\mathfrak h$ on a separable Hilbert
space $\cH_W$, and set $\mathcal R_W(K):=\{W(f):f\in K\}''$.  If
$K_1,K_2\subset\mathfrak h$ are closed real subspaces with $\omega(K_1,K_2)\ne0$, then there is
a self-adjoint unitary $B\in\mathcal R_W(K_2)$ such that for every normal state $\rho$ on
$B(\cH_W)$ one can choose a self-adjoint unitary $A_\rho\in\mathcal R_W(K_1)$ with
\begin{equation}
\Var_\rho\!\left(-i[A_\rho,B]\right)=4.
\label{eq:intro-statewise}
\end{equation}
Consequently, for every positive self-adjoint operator $G$ on $\cH_W$ and every
$E>\inf\sigma(G)$,
\begin{equation}
\Gamma_{G,E}(\mathcal R_W(K_1),\mathcal R_W(K_2))=4.
\label{eq:intro-saturation}
\end{equation}
\end{theorem}

The proof uses a phase-adapted Clifford family.  A nonzero Weyl pairing can be scaled to
phase $\pi$, so that two Weyl unitaries anticommute.  An antipodal Borel sign function produces
self-adjoint unitary Clifford generators, and a scalar phase rotation gives a strongly
continuous family of them.  For any normal state, the commutator expectation is a continuous
real antiperiodic function of the phase and hence vanishes for some choice of phase.  The square
of the commutator is identically $4$, so the variance is then maximal.  The energy-constrained
statement follows by choosing any state in the nonempty set $\mathfrak S_{G,E}$.

The same family yields the exact finite-parameter constrained Clifford slice in
Corollary~\ref{cor:constrained-clifford-slice}.  By contrast, a fixed energy-regular Weyl
generator retains a genuine dependence on the relation between the reference energy and the
one-particle dynamics.  We therefore consider the trigonometric Weyl contractions
\begin{equation}
A=\frac{W(u)-W(-u)}{2i},
\qquad
B=\frac{W(v)-W(-v)}{2i},
\qquad
\omega(u,v)=\pi,
\label{eq:trig-intro}
\end{equation}
and $h=-i[A,B]$.  These are obtained by continuous functional calculus.  If
$u,v\in\dom H_1^{1/2}$, these generators are bounded on the $G^{1/2}$-graph space for
$G=d\Gamma(H_1)$.

The fixed witness also has an exact algebraic structure.  In one canonical mode write
$R=(Q,P)^T$, $[Q,P]=i$, and
\[
\Omega=\begin{pmatrix}0&1\\-1&0\end{pmatrix},
\qquad
W(z)=e^{-iz^T\Omega R}.
\]
For $M>0$ set
\[
G_M=\frac12R^TMR-\frac12\sqrt{\det M}.
\]
If $u^T\Omega v=\pi$, let $h$ be the trigonometric witness in
\eqref{eq:trig-intro}.

\begin{theorem}[Sharp fixed-witness asymptotic for a quadratic energy]
\label{thm:intro-quadratic-bridge}
With the preceding notation,
\begin{equation}
4-\gamma_{G_M,E}(h)
=\frac{\kappa_M(u,v)}{E}+O(E^{-2}),
\qquad E\to\infty,
\label{eq:intro-quadratic-bridge}
\end{equation}
where
\begin{equation}
\kappa_M(u,v)
=\frac14\Bigl(u^TMu+v^TMv+2\pi\sqrt{\det M}\Bigr).
\label{eq:intro-kappa}
\end{equation}
For $M=I_2$ and the square pair
$u=(\sqrt\pi,0)^T$, $v=(0,\sqrt\pi)^T$, the coefficient is $\pi$.
\end{theorem}

The proof starts from the identity for $4-h^2$, which produces a phase-fixed stabilizer
Hamiltonian for the squared Weyl pair and its two product directions.  A translated-unitary
uncertainty estimate gives a converse over all normal states satisfying the mean quadratic-energy
bound, while a localized Zak Gaussian reaches the same coefficient.  A finite-energy
symmetrization removes the mean of $h$ at a constant energy cost and transfers the stabilizer
asymptotic to the commutator variance.

The argument uses three established ingredients in a different variational setting.
Breitenberger-type unitary uncertainty estimates control translations by a conjugate observable
\cite{Breitenberger1985}; finite-energy GKP constructions exhibit inverse-energy stabilizer
scaling \cite{GottesmanKitaevPreskill2001,MatsuuraYamasakiKoashi2020}; and the lattice structure
of multimode GKP stabilizers is naturally described in symplectic terms
\cite{ConradEisertArzani2022}.  Zak coordinates are well suited to the periodic Weyl pair used in
the matching construction \cite{EnglertZak2006,PantaleoniZak2023}.  Here the additional step is an all-state converse for the fixed commutator witness, matched by a
construction attaining the same coefficient.  The coefficient concerns this stabilizer/commutator diagnostic; stabilizer
expectations alone do not constitute a fidelity guarantee for a GKP code state
\cite{Goldberg2025Stabilizers}.

For general second-quantized energies the same mechanism gives a two-sided inverse-energy rate, although there is
no universal finite-dimensional coefficient of the form \eqref{eq:intro-kappa}.

\begin{theorem}[General second-quantized fixed-witness dichotomy]
\label{thm:intro-regular}
Let $H_1\ge0$ be self-adjoint on $\mathfrak h$, let $G=d\Gamma(H_1)$, and let
$u,v\in\dom H_1^{1/2}$ satisfy $\omega(u,v)=\pi$.  For the witness $h$ in
\eqref{eq:trig-intro} there are $C,E_0<\infty$ such that
\begin{equation}
0\le4-\gamma_{G,E}(h)\le\frac{C}{E},
\qquad E\ge E_0.
\label{eq:intro-upper-fixed}
\end{equation}
If at least one of $u,v$ lies outside $\ker H_1$, then there are $c,C,E_0>0$ such that
\begin{equation}
\frac{c}{E+1}\le4-\gamma_{G,E}(h)\le\frac{C}{E},
\qquad E\ge E_0.
\label{eq:intro-sharp-fixed}
\end{equation}
If instead $u,v\in\ker H_1$, then
\begin{equation}
\gamma_{G,E}(h)=4
\qquad\text{for every }E>0.
\label{eq:intro-zero-mode}
\end{equation}
Thus the fixed witness has two energy regimes: an inverse-energy deficit when at least one direction
carries positive one-particle energy, and exact saturation when both directions are zero modes.
Injectivity of $H_1$ is a simple sufficient condition for the first regime.
\end{theorem}

For a general $d\Gamma(H_1)$, the upper bound is obtained from a compact finite comb that
approximates the two commuting squared Weyl translations with error $O(R^{-1})$ at number-energy
budget $R$.  A finite-dimensional symplectic transformation transports the construction to the
given pair, and the one-particle form of $H_1$ converts number energy into physical energy.  The
matching lower rate follows from the translated-unitary uncertainty estimate together with an
inverse-one-particle-energy bound for a conjugate field quadrature.  Since
$\dom H_1^{-1/2}$ is dense in $(\ker H_1)^\perp$, such a quadrature can be chosen whenever the
witness is not supported entirely on zero modes.

The contrast is between adapting the bounded witness to the state and keeping an
energy-regular witness fixed.  The Clifford symmetries in
Theorem~\ref{thm:intro-saturation} come from discontinuous Borel functional calculus and can
fail to preserve the number form domain in the oscillator representation.  The trigonometric
generators in Theorem~\ref{thm:intro-regular}, by contrast, obey explicit graph-energy
estimates.

Section~2 fixes the Weyl and constrained-variance notation.  Section~3 proves statewise
saturation and records the regularity distinction between the phase-selected Borel symmetries and
the fixed trigonometric generators.  Section~4 derives the squared-stabilizer identity and the
general inverse-energy upper bound.  Section~5 determines the sharp coefficient for one-mode
quadratic energies.  Section~6 proves the lower rate for general second-quantized energies and the
zero-mode alternative.

\section{Weyl representations and constrained variance}

Let $\mathfrak h$ be a separable complex Hilbert space with symplectic form
\[
\omega(f,g)=\operatorname{Im}\langle f,g\rangle.
\]
Let $W:\mathfrak h\to\mathcal U(\cH_W)$ be a representation of the Weyl relations on a
separable Hilbert space $\cH_W$, with convention
\begin{equation}
W(f)W(g)=e^{-i\omega(f,g)/2}W(f+g).
\label{eq:weyl}
\end{equation}
For a closed real linear subspace $K\subset\mathfrak h$, set
\[
\mathcal R_W(K):=\{W(f):f\in K\}'',
\qquad
K':=\{g\in\mathfrak h:\omega(f,g)=0\ \text{for every }f\in K\}.
\]

For a positive self-adjoint $G$ on a Hilbert space $\cH$ and
$E>\inf\sigma(G)$, let $\mathfrak S_{G,E}$ be as in \eqref{eq:gamma-intro},
with $\Tr(G\rho)$ understood in the usual extended-positive sense.  The quantity
$\gamma_{G,E}(h)$ is well-defined for every bounded self-adjoint $h$, and
\begin{equation}
\gamma_{G,E}(h)\le2\|h\|.
\label{eq:gamma-bound}
\end{equation}
Indeed $\Var_\rho(h)\le\Tr(\rho h^2)\le\|h\|^2$.

For later use, if $\Delta$ is the difference of two channels on the trace class of $\cH_W$,
define its mean-input-energy-constrained diamond norm by
\begin{equation}
\|\Delta\|_{\diamond,G,E}
:=
\sup_{\substack{\mathcal K,\ \rho\ge0,\ \Tr\rho=1\\
\rho\in\mathcal T(\cH_W\otimes\mathcal K),\
\Tr[(G\otimes1)\rho]\le E}}
\bigl\|(\Delta\otimes\id_{\mathcal K})(\rho)\bigr\|_1,
\label{eq:EC-diamond}
\end{equation}
where the supremum is over separable reference Hilbert spaces $\mathcal K$ and the energy is understood
in the extended-positive sense.
This is the standard energy-constrained diamond norm
\cite{Shirokov2018,BeckerDattaLamiRouze2021}.

We shall use the following immediate consequence of the Weyl relations.

\begin{lemma}[Anticommuting Weyl pair]
\label{lem:anti-weyl}
If $u,v\in\mathfrak h$ satisfy $\omega(u,v)=\pi$ and
$U=W(u)$, $V=W(v)$, then
\[
UV=-VU,
\qquad
VUV^*=-U.
\]
\end{lemma}

\begin{proof}
The first identity is the exchange form of \eqref{eq:weyl}; the second follows by conjugation
with $V$.
\end{proof}

\section{Statewise and energy-constrained saturation}

The statewise result rests on a phase-selection argument that is independent of the particular Weyl realization.

\begin{proposition}[Phase-selection criterion]
\label{prop:phase-selection}
Let $M,N\subset B(\cH)$ contain self-adjoint unitaries $B\in N$ and
$A_\alpha\in M$, $\alpha\in\mathbb R$, such that
\begin{equation}
A_\alpha B=-BA_\alpha,
\qquad
A_{\alpha+\pi}=-A_\alpha,
\label{eq:abstract-clifford-family}
\end{equation}
and suppose $\alpha\mapsto A_\alpha$ is strongly continuous.  Then for every normal state
$\rho$ on $B(\cH)$ there is $\alpha_\rho\in\mathbb R$ such that
\begin{equation}
\Var_\rho\!\left(-i[A_{\alpha_\rho},B]\right)=4.
\label{eq:phase-selection}
\end{equation}
\end{proposition}

\begin{proof}
Set
\[
h_\alpha=-i[A_\alpha,B]=-2iA_\alpha B.
\]
Because $A_\alpha$ and $B$ are anticommuting symmetries, $h_\alpha^2=4\,1$.
For a normal state $\rho$, let
\[
F_\rho(\alpha)=\Tr(\rho h_\alpha).
\]
The function is real and antiperiodic.  It is continuous as well: strong continuity of
$A_\alpha$, together with the uniform bound $\|h_\alpha\|\le2$, gives continuity of the trace
pairing, for example by approximating the density operator of $\rho$ in trace norm by
finite-rank operators.  Hence $F_\rho(\alpha_\rho)=0$ for some $\alpha_\rho$.  Therefore
\[
\Var_\rho(h_{\alpha_\rho})
=
\Tr(\rho h_{\alpha_\rho}^2)-F_\rho(\alpha_\rho)^2
=4.
\]
\end{proof}

It remains to construct such a family from a nonzero Weyl pairing.  The first step excludes
point spectrum.

\begin{lemma}[Diffuse spectrum of a nontrivial Weyl unitary]
\label{lem:no-eigenvalues}
If $u\ne0$, then $W(u)$ has no eigenvalues on $\cH_W$.
\end{lemma}

\begin{proof}
Since $u\ne0$, one may take $w=iu$, for which $\omega(u,w)=\|u\|^2\ne0$.  If $W(u)\psi=\lambda\psi$ for a nonzero
vector $\psi$, then the Weyl relations imply
\[
W(u)W(tw)\psi
=e^{-it\omega(u,w)}\lambda\,W(tw)\psi,
\qquad t\in\mathbb R.
\]
As $t$ varies over an interval of length less than
$2\pi/|\omega(u,w)|$, the eigenvalues on the right are uncountably many and distinct.  The
corresponding vectors are mutually orthogonal because they belong to distinct eigenspaces of the
same unitary.  This contradicts separability of $\cH_W$.
\end{proof}

Choose once and for all a Borel function $\chi:\mathbb T\to\{-1,1\}$ with
\begin{equation}
\chi(-z)=-\chi(z)
\qquad(z\in\mathbb T)
\label{eq:chi-odd}
\end{equation}
and only two discontinuities.  For example, let
\[
B_+:=\{z\in\mathbb T:\operatorname{Re}z>0\}\cup\{i\},
\qquad
\chi=1_{B_+}-1_{-B_+}.
\]

\begin{lemma}[Strongly continuous Clifford family]
\label{lem:strong-clifford}
Let $u,v\in\mathfrak h$ satisfy $\omega(u,v)=\pi$, put $U=W(u)$, $V=W(v)$, and define
\begin{equation}
A_\alpha:=\chi(e^{-i\alpha}U),
\qquad
B:=\chi(V).
\label{eq:Aalpha}
\end{equation}
Then $A_\alpha$ and $B$ are self-adjoint unitaries,
$\alpha\mapsto A_\alpha$ is strongly continuous, and
\begin{equation}
A_\alpha B=-BA_\alpha,
\qquad
A_{\alpha+\pi}=-A_\alpha.
\label{eq:weyl-clifford-family}
\end{equation}
\end{lemma}

\begin{proof}
Self-adjoint unitarity and the second relation in \eqref{eq:weyl-clifford-family} follow from
Borel functional calculus and \eqref{eq:chi-odd}.  To prove strong continuity, let $E_U$ be the
spectral measure of $U$ and
$\mu_\psi(\cdot)=\langle\psi,E_U(\cdot)\psi\rangle$.  By
Lemma~\ref{lem:no-eigenvalues}, $\mu_\psi$ has no atoms.  Therefore, if $\alpha_n\to\alpha$,
\[
\|(A_{\alpha_n}-A_\alpha)\psi\|^2
=
\int_{\mathbb T}
|\chi(e^{-i\alpha_n}z)-\chi(e^{-i\alpha}z)|^2\,d\mu_\psi(z)
\longrightarrow0
\]
by dominated convergence: pointwise convergence fails only at the two rotated discontinuities
of $\chi$.

By Lemma~\ref{lem:anti-weyl}, $VUV^*=-U$.  Hence
\[
VA_\alpha V^*=\chi(-e^{-i\alpha}U)=-A_\alpha,
\]
so $A_\alpha V A_\alpha=-V$.  Applying Borel functional calculus once more,
\[
A_\alpha B A_\alpha
=\chi(A_\alpha V A_\alpha)
=\chi(-V)
=-B,
\]
which proves the first relation in \eqref{eq:weyl-clifford-family}.
\end{proof}

\begin{proof}[Proof of Theorem~\ref{thm:intro-saturation}]
Choose $f\in K_1$ and $g\in K_2$ with $\omega(f,g)\ne0$, and rescale by real numbers so that
$u\in K_1$, $v\in K_2$ satisfy $\omega(u,v)=\pi$.  Lemma~\ref{lem:strong-clifford} supplies
the family required by Proposition~\ref{prop:phase-selection}, with
$M=\mathcal R_W(K_1)$ and $N=\mathcal R_W(K_2)$, proving
\eqref{eq:intro-statewise}.

For the energy-constrained statement, choose $E'$ with
$\inf\sigma(G)<E'<E$.  The spectral projection
$P_{E'}:=1_{(-\infty,E']}(G)$ is nonzero, so any unit vector in its range defines an admissible
state in $\mathfrak S_{G,E}$.  Apply the statewise result to such a state.  The resulting witness satisfies
$\gamma_{G,E}(-i[A_\rho,B])=4$ by \eqref{eq:gamma-bound}, and the universal bound
$\Gamma_{G,E}\le4$ gives
\eqref{eq:intro-saturation}.
\end{proof}

\begin{corollary}[Exact constrained Clifford slice]
\label{cor:constrained-clifford-slice}
Let $K_1,K_2\subset\mathfrak h$ satisfy the hypotheses of
Theorem~\ref{thm:intro-saturation}, let $G$ be positive self-adjoint on $\cH_W$, and let
$E>\inf\sigma(G)$.  Then there exist self-adjoint unitaries
$A_E\in\mathcal R_W(K_1)$ and $B\in\mathcal R_W(K_2)$ and a unit vector $\psi_E$ with
\[
\langle\psi_E,G\psi_E\rangle<E
\]
such that the following holds.  Let
\[
\widehat\alpha_s(\rho)=e^{isA_E}\rho e^{-isA_E},
\qquad
\widehat\beta_t(\rho)=e^{itB}\rho e^{-itB},
\]
and let
\[
C^{G,E}_{A_E,B}(s,t)
:=
\bigl\|
\widehat\beta_t\widehat\alpha_s-\widehat\alpha_s\widehat\beta_t
\bigr\|_{\diamond,G,E},
\]
with the norm defined in \eqref{eq:EC-diamond}.  Then, for every $s\in\mathbb R$,
\begin{equation}
C^{G,E}_{A_E,B}\left(s,\frac{\pi}{2}\right)
=
2|\sin 2s|.
\label{eq:constrained-Clifford-slice}
\end{equation}
In particular,
\begin{equation}
C^{G,E}_{A_E,B}\left(\frac{\pi}{4},\frac{\pi}{2}\right)=2,
\label{eq:perfect-order-discrimination}
\end{equation}
so, for equal prior probabilities, the two resulting output states are perfectly
distinguishable in one use.  The constraint here is on the input state; no energetic cost is
imposed on the implementation of the Borel symmetry $A_E$.
\end{corollary}

\begin{proof}
Choose $E'$ with $\inf\sigma(G)<E'<E$ and a unit vector
$\psi_E\in1_{(-\infty,E']}(G)\cH_W$.  After scaling a nonzero pairing to
$\omega(u,v)=\pi$, Lemma~\ref{lem:strong-clifford} gives a fixed symmetry
$B\in\mathcal R_W(K_2)$ and a strongly continuous family of symmetries
$A_\alpha\in\mathcal R_W(K_1)$ satisfying
\[
A_\alpha B=-BA_\alpha,
\qquad
A_{\alpha+\pi}=-A_\alpha.
\]
The function
\[
f(\alpha)=\langle\psi_E,A_\alpha\psi_E\rangle
\]
is real, continuous, and antiperiodic.  Hence there is $\alpha_E$ with
$f(\alpha_E)=0$.  Put $A_E=A_{\alpha_E}$.

For the two pure output vectors at $t=\pi/2$,
\[
\psi_{BA}(s)=e^{i\pi B/2}e^{isA_E}\psi_E,
\qquad
\psi_{AB}(s)=e^{isA_E}e^{i\pi B/2}\psi_E,
\]
anticommutation gives
\[
\langle\psi_{BA}(s),\psi_{AB}(s)\rangle
=
\left\langle\psi_E,
e^{-isA_E}Be^{isA_E}B\psi_E
\right\rangle
=
\langle\psi_E,e^{-2isA_E}\psi_E\rangle
=
\cos 2s.
\]
Therefore
\[
\left\|
|\psi_{BA}(s)\rangle\langle\psi_{BA}(s)|
-
|\psi_{AB}(s)\rangle\langle\psi_{AB}(s)|
\right\|_1
=
2|\sin 2s|.
\]
Since $\psi_E$ satisfies the input-energy constraint, this gives the lower bound
\[
C^{G,E}_{A_E,B}\left(s,\frac{\pi}{2}\right)\ge2|\sin2s|.
\]
For the reverse inequality, let $\varrho$ be any admissible system--reference state and, after
enlarging the reference system if necessary, purify it to a unit vector $\Xi$ without changing
the system marginal.  The two purified outputs differ by
the relative unitary $e^{-2isA_E}\otimes1$, so with
\[
a_\Xi:=\langle\Xi,(A_E\otimes1)\Xi\rangle\in[-1,1]
\]
their overlap is
\[
\langle\Xi,(e^{-2isA_E}\otimes1)\Xi\rangle
=
\cos2s-i\,a_\Xi\sin2s.
\]
Its modulus is at least $|\cos2s|$.  Hence the trace distance of the purified outputs is at most
$2|\sin2s|$, and trace-norm contractivity under the partial trace gives the same bound for the
original outputs.  Taking the supremum over admissible inputs proves
\eqref{eq:constrained-Clifford-slice}.  At $s=\pi/4$ the two output states built from $\psi_E$
are orthogonal, and the final statement follows from the Helstrom formula.
\end{proof}

The finite-dimensional unconstrained unitary-discrimination formula is classical
\cite[Thm.~3.55]{Watrous2018TQI}, while energy-constrained discrimination of unitary channels is
studied in \cite{BeckerDattaLamiRouze2021}.  Here the Clifford structure makes the optimum explicit
for an input below any prescribed admissible energy threshold.

\begin{corollary}[Energy-independent Weyl dichotomy]
\label{cor:weyl-dichotomy}
For closed real subspaces $K_1,K_2\subset\mathfrak h$, any positive self-adjoint $G$ on
$\cH_W$, and $E>\inf\sigma(G)$,
\begin{equation}
\Gamma_{G,E}(\mathcal R_W(K_1),\mathcal R_W(K_2))
=
\begin{cases}
0,&\omega(K_1,K_2)=0,\\
4,&\omega(K_1,K_2)\ne0.
\end{cases}
\label{eq:weyl-dichotomy}
\end{equation}
\end{corollary}

\begin{proof}
If $\omega(K_1,K_2)=0$, the Weyl relations show that $W(f)$ and $W(g)$ commute for every
$f\in K_1$ and $g\in K_2$.  Hence
$[\mathcal R_W(K_1),\mathcal R_W(K_2)]=0$, which gives the first line.  The second is
Theorem~\ref{thm:intro-saturation}.
\end{proof}

\begin{remark}[State dependence and energy regularity]
\label{rem:state-dependence}
The phase in Theorem~\ref{thm:intro-saturation} cannot be fixed independently of the state.  If
$k=-i[a,b]$ with $\|a\|,\|b\|\le1$ and $\Var_\rho(k)=4$, then $\|k\|=2$, the state $\rho$
is supported on the spectral subspaces of $k$ at $\pm2$, and its mean-zero condition puts
nonzero weight on both.  A normal state supported on either endpoint subspace has variance zero.
Thus a single admissible pair cannot have variance $4$ in every normal state.  The phase-selected
Borel symmetries may also fail to preserve the energy form domain, as the following oscillator
example shows.
\end{remark}

\begin{proposition}[A phase-selected symmetry outside the number form domain]
\label{prop:singular-symmetry}
Let $\cH=L^2(\mathbb R)$ with oscillator number operator
$N=(Q^2+P^2-1)/2$, and let $U=e^{iaQ}$ with $a\ne0$.  For the antipodal sign function $\chi$
used above and any $\alpha\in\mathbb R$, set $A_\alpha=\chi(e^{-i\alpha}U)$.  If $\Omega$ is
the Gaussian oscillator vacuum, then
\begin{equation}
A_\alpha\Omega\notin\dom N^{1/2}.
\label{eq:singular-symmetry}
\end{equation}
Consequently, if $t\notin\pi\mathbb Z$, the vector $e^{itA_\alpha}\Omega$ has infinite number
form energy.
\end{proposition}

\begin{proof}
The operator $A_\alpha$ is multiplication by a $\{\pm1\}$-valued periodic step function whose
jumps occur where $e^{i(ax-\alpha)}=\pm i$.  The oscillator vacuum is smooth and nonzero at
those points, so $A_\alpha\Omega$ has jump discontinuities and does not belong to $H^1(\mathbb R)$.
The form domain of $N+1/2$ is
\[
\{\psi\in L^2(\mathbb R):\psi'\in L^2(\mathbb R),\ x\psi\in L^2(\mathbb R)\},
\]
which proves \eqref{eq:singular-symmetry}.  Since
$e^{itA_\alpha}\Omega=\cos t\,\Omega+i\sin t\,A_\alpha\Omega$, the same jump discontinuities
remain when $\sin t\ne0$.
\end{proof}

\section{Fixed trigonometric witnesses and finite-energy approximation}

From this point onward we specialize to the symmetric Fock representation on
$\cF(\mathfrak h)$ and write
\[
\Ralg(K):=\{W(f):f\in K\}''.
\]
For $w\in\mathfrak h$, let $\Phi(w)$ denote the Segal field normalized by
$W(tw)=e^{it\Phi(w)}$.

We record the following graph-energy estimate from \cite[Lem.~5.11]{Nasreddine2026DualityArxiv}, including the form-domain argument because it will be used below.

\begin{lemma}[Weyl graph-energy bound]
\label{lem:weyl-energy}
Let $H_1\ge0$ be self-adjoint on $\mathfrak h$, put $G=d\Gamma(H_1)$, and let
$z\in\dom H_1^{1/2}$.  Then, for $\Psi\in\dom G^{1/2}$,
\begin{equation}
\|G^{1/2}W(z)\Psi\|
\le
\|G^{1/2}\Psi\|+\frac{1}{\sqrt2}\|H_1^{1/2}z\|\,\|\Psi\|.
\label{eq:weyl-energy}
\end{equation}
\end{lemma}

\begin{proof}
Choose an orthonormal basis $(e_j)$ contained in $\dom H_1$.  On the usual finite-particle
form core for $G$,
\[
\|G^{1/2}\Psi\|^2
=\sum_j\|a(H_1^{1/2}e_j)\Psi\|^2.
\]
The Weyl displacement relation shifts each annihilation operator by a scalar whose modulus is
$|\langle e_j,H_1^{1/2}z\rangle|/\sqrt2$.  Minkowski's inequality in the Hilbert direct
sum and Parseval's identity give \eqref{eq:weyl-energy} on the form core.

For $\varepsilon>0$, set
\[
H_{1,\varepsilon}:=H_1(1+\varepsilon H_1)^{-1},
\qquad
G_\varepsilon:=d\Gamma(H_{1,\varepsilon}).
\]
Since $H_{1,\varepsilon}$ is bounded, $G_\varepsilon\le\|H_{1,\varepsilon}\|N$ as quadratic
forms, and the standard Weyl number estimate shows that $W(z)$ preserves $\dom N^{1/2}$.
Thus the preceding computation is legitimate on the finite-particle form core, and graph-norm
closure for $G_\varepsilon^{1/2}$ gives
\[
\|G_\varepsilon^{1/2}W(z)\Psi\|
\le
\|G_\varepsilon^{1/2}\Psi\|
+\frac1{\sqrt2}\|H_{1,\varepsilon}^{1/2}z\|\,\|\Psi\|
\]
for $\Psi\in\dom G^{1/2}\subset\dom G_\varepsilon^{1/2}$.  As
$\varepsilon\downarrow0$, the forms of $G_\varepsilon$ increase to that of $G$, while
$H_{1,\varepsilon}^{1/2}z\to H_1^{1/2}z$.  Monotone convergence of the quadratic forms
therefore gives $W(z)\Psi\in\dom G^{1/2}$ and \eqref{eq:weyl-energy}.
\end{proof}

Averaging \eqref{eq:weyl-energy} for $W(\pm u)$ and $W(\pm v)$ shows that the trigonometric
generators below preserve $\dom G^{1/2}$ and are bounded for the $G^{1/2}$ graph norm whenever
$u,v\in\dom H_1^{1/2}$.

Fix $u,v\in\mathfrak h$ with
$\omega(u,v)=\pi$ and put
\begin{equation}
U=W(u),\qquad V=W(v),\qquad
A=\frac{U-U^*}{2i},\qquad
B=\frac{V-V^*}{2i},\qquad
h=-i[A,B].
\label{eq:trig-witness}
\end{equation}
Then $A$ and $B$ are self-adjoint contractions, $AB=-BA$, and
\begin{equation}
h=-2iAB,
\qquad
h^2=4A^2B^2.
\label{eq:h-square}
\end{equation}
Set
\begin{equation}
z_1=2u,\qquad z_2=2v,\qquad
S_1=-W(z_1),\qquad S_2=-W(z_2).
\label{eq:squared-stabilizers}
\end{equation}
Since $\omega(z_1,z_2)=4\pi$, the unitaries $S_1$ and $S_2$ commute.

\begin{proposition}[Squared-stabilizer identity]
\label{prop:squared-stabilizer}
For the fixed witness \eqref{eq:trig-witness},
\begin{align}
4-h^2
={}&(1-\operatorname{Re}S_1)+(1-\operatorname{Re}S_2) \notag\\
&+\frac12(1-\operatorname{Re}S_1S_2)
+\frac12(1-\operatorname{Re}S_1S_2^*).
\label{eq:squared-stabilizer-identity}
\end{align}
Equivalently, $2(4-h^2)$ is the phase-fixed stabilizer Hamiltonian associated with the
multiset
\begin{equation}
\mathcal Z_h=(z_1,z_1,z_2,z_2,z_1+z_2,z_1-z_2)
\label{eq:bridge-frame}
\end{equation}
and target phases $(-1,-1,-1,-1,+1,+1)$.
\end{proposition}

\begin{proof}
Write $X=\operatorname{Re}S_1$ and $Y=\operatorname{Re}S_2$.  From
\eqref{eq:squared-stabilizers},
\[
A^2=\frac{1+X}{2},\qquad B^2=\frac{1+Y}{2},
\]
so \eqref{eq:h-square} gives $h^2=(1+X)(1+Y)$.  Because $S_1$ and $S_2$ commute,
\[
XY=\frac12\operatorname{Re}(S_1S_2)+\frac12\operatorname{Re}(S_1S_2^*),
\]
which proves \eqref{eq:squared-stabilizer-identity}.  The Weyl cocycle is trivial on
$z_1,z_2$ because their symplectic pairing is $4\pi$, hence
$S_1S_2=W(z_1+z_2)$ and $S_1S_2^*=W(z_1-z_2)$.  This gives the stated phase assignment.
\end{proof}

The product directions in \eqref{eq:bridge-frame} contribute at the same order as the two
squared generators.  Thus the fixed-witness deficit is not determined by $S_1$ and $S_2$
separately.  Together with Remark~\ref{rem:state-dependence} and
Proposition~\ref{prop:singular-symmetry}, the identity also separates the two regimes considered
in the paper: exact saturation is obtained by adapting a Borel symmetry to the state, whereas the
fixed trigonometric witness remains energy regular and its approach to the endpoint is controlled
by the commuting frame in \eqref{eq:bridge-frame}.

\begin{lemma}[Spectral diameter of the fixed witness]
\label{lem:h-spectral-diameter}
For the trigonometric witness \eqref{eq:trig-witness},
\[
\diam\sigma(h)=4.
\]
\end{lemma}

\begin{proof}
The restriction of the regular Weyl representation to the real symplectic plane
$\operatorname{span}_{\mathbb R}\{u,v\}$ is, by the Stone--von Neumann theorem, a multiple of the
Schr\"odinger representation; the multiplicity does not change the spectrum.  In that realization
$(S_1,S_2)$ is, up to constant phases, a translation and a modulation with cell area $4\pi$.  In the
Zak representation these become two multiplication operators whose joint essential range is
$\mathbb T^2$.  Hence
$(1,1)$ belongs to the joint spectrum, and \eqref{eq:h-square} gives $\|h\|=2$.  Conjugation by
$V=W(v)$ sends $A$ to $-A$ and fixes $B$, so $VhV^*=-h$.  The spectrum of $h$ is therefore
symmetric about zero, and its diameter is $4$.
\end{proof}

The next estimates quantify how finite input energy approaches this spectral endpoint when
$u,v$ have finite one-particle form energy.

\subsection{A compact comb for commuting Weyl translations}

Let $Q$ and $P=-i\,d/dx$ be the canonical operators on $L^2(\mathbb R)$ and
\[
N=\frac{Q^2+P^2-1}{2}.
\]
For $L,K>0$ set
\[
(T_L\psi)(x)=\psi(x-L),
\qquad
(M_K\psi)(x)=e^{iKx}\psi(x).
\]

\begin{lemma}[Compact comb]
\label{lem:compact-comb}
Assume $KL=4\pi$.  Let $\zeta_1,\zeta_2\in\mathbb T$ and
$S_1=\zeta_1T_L$, $S_2=\zeta_2M_K$.  There are $C,R_0<\infty$ such that for every
$R\ge R_0$ one can find a unit vector $\psi_R\in\dom N^{1/2}$ satisfying
\begin{equation}
\langle\psi_R,N\psi_R\rangle\le R,
\qquad
1-\operatorname{Re}\langle\psi_R,S_j\psi_R\rangle\le\frac{C}{R},
\quad j=1,2.
\label{eq:comb-main}
\end{equation}
\end{lemma}

\begin{proof}
Choose a real even $\varphi\in C_c^\infty(\mathbb R)$ with $\|\varphi\|_2=1$ and
$\operatorname{supp}\varphi\subset(-r_0,r_0)$.  Choose $q_0$ so that
$\zeta_2e^{iKq_0}=1$.  For an integer $J\ge2$ define
\[
q_{n,J}=q_0+\left(n-\frac{J+1}{2}\right)L,
\qquad
\varphi_{\sigma,n}(x)
=\sigma^{-1/2}\varphi\left(\frac{x-q_{n,J}}{\sigma}\right),
\]
and
\[
d_n=\sqrt{\frac{2}{J+1}}\sin\frac{\pi n}{J+1}.
\]
For $2r_0\sigma<L$ the supports of the $\varphi_{\sigma,n}$ are disjoint.  Choose a phase
$\theta$ so that the nearest-neighbor contribution of $S_1$ below is positive and set
\[
\psi_{\sigma,J}=\sum_{n=1}^J e^{in\theta}d_n\varphi_{\sigma,n}.
\]
The vector is normalized.  Since $T_L\varphi_{\sigma,n}=\varphi_{\sigma,n+1}$,
\[
\operatorname{Re}\langle\psi_{\sigma,J},S_1\psi_{\sigma,J}\rangle
=\sum_{n=1}^{J-1}d_nd_{n+1}
=\cos\frac{\pi}{J+1}.
\]
Hence
\begin{equation}
1-\operatorname{Re}\langle S_1\rangle_{\psi_{\sigma,J}}
\le\frac{\pi^2}{2(J+1)^2}.
\label{eq:comb-S1}
\end{equation}
Because $KL=4\pi$, $e^{iKq_{n,J}}$ is independent of $n$ and of the parity of $J$; the
factor $4\pi$ removes the parity-dependent global phase of the centered lattice.  The choice of
$q_0$ and evenness
of $|\varphi|^2$ give
\[
\operatorname{Re}\langle S_2\rangle_{\psi_{\sigma,J}}
=\int_{\mathbb R}|\varphi(y)|^2\cos(K\sigma y)\,dy,
\]
so
\begin{equation}
1-\operatorname{Re}\langle S_2\rangle_{\psi_{\sigma,J}}
\le
\frac{K^2\sigma^2}{2}
\int_{\mathbb R}y^2|\varphi(y)|^2\,dy.
\label{eq:comb-S2}
\end{equation}
Disjointness of the supports gives
\[
\langle P^2\rangle_{\psi_{\sigma,J}}
=\sigma^{-2}\|\varphi'\|_2^2,
\]
and, writing $m_2=\int y^2|\varphi(y)|^2dy$,
\[
\langle Q^2\rangle_{\psi_{\sigma,J}}
=\sum_{n=1}^Jd_n^2q_{n,J}^2+\sigma^2m_2
\le\left(|q_0|+\frac{LJ}{2}\right)^2+\sigma^2m_2.
\]
The choice $J\asymp\sigma^{-1}$ balances the translation error $J^{-2}$ with the modulation
error $\sigma^2$, while both position and momentum contributions to the number energy are
$O(\sigma^{-2})$.  Take $J=\lceil\sigma^{-1}\rceil$.  For sufficiently small $\sigma$ there are constants
$a,b,c>0$, independent of $\sigma$, such that
\[
\langle N\rangle_{\psi_{\sigma,J}}\le a\sigma^{-2}+b,
\qquad
1-\operatorname{Re}\langle S_j\rangle_{\psi_{\sigma,J}}\le c\sigma^2.
\]
For $R$ large choose $\sigma^2=2a/R$.  Then $\langle N\rangle\le R$ after increasing the
threshold, while the stabilizer errors are at most $2ac/R$.
\end{proof}

The construction is related to the finite-energy lattice states used in continuous-variable
coding \cite{GottesmanKitaevPreskill2001,MatsuuraYamasakiKoashi2020}.  The compact support and
sine envelope are useful here because the overlap and number-energy estimates are both explicit,
which allows the construction to be inserted directly into the commutator-variance problem.

\subsection{Transfer to a finite-energy Weyl pair}

\begin{lemma}[Finite-dimensional symplectic transfer]
\label{lem:symplectic-transfer}
Let $z_1,z_2\in\mathfrak h$ satisfy $\omega(z_1,z_2)=4\pi$, and put
$\mathfrak k=\spn_{\mathbb C}\{z_1,z_2\}$.  There are $C,R_0<\infty$ such that for every
$R\ge R_0$ there is a unit vector $\Psi_R\in\cF(\mathfrak k)$ with
\begin{equation}
\langle\Psi_R,N_{\mathfrak k}\Psi_R\rangle\le R,
\qquad
1-\operatorname{Re}\langle\Psi_R,-W(z_j)\Psi_R\rangle\le\frac{C}{R},
\quad j=1,2,
\label{eq:finite-mode}
\end{equation}
where $N_{\mathfrak k}=d\Gamma(1_{\mathfrak k})$.
\end{lemma}

\begin{proof}
The real symplectic space underlying $\mathfrak k$ is finite dimensional.  Choose a canonical
symplectic pair $z_1^{(0)},z_2^{(0)}$ with pairing $4\pi$, supported on one canonical mode.
In its Schr\"odinger realization the corresponding Weyl operators are, up to scalar phases, a
translation and a modulation with $KL=4\pi$.  Lemma~\ref{lem:compact-comb}, tensored with the
vacuum on any remaining modes, therefore supplies canonical vectors with number budget $R'$
and stabilizer error $O((R')^{-1})$.

The symplectic isomorphism sending $z_j^{(0)}$ to $z_j$ extends, by finite-dimensional
symplectic Gram--Schmidt, to a real symplectic automorphism $T$ of $\mathfrak k$.  It is
unitarily implementable on $\cF(\mathfrak k)$; write
$\mathcal U_TW(z)\mathcal U_T^*=W(Tz)$ \cite{Shale1962}.
Let $\mathbf R=(Q_1,P_1,\ldots,Q_d,P_d)^T$ be canonical quadratures, so that
\[
N_{\mathfrak k}+\frac d2=\frac12\mathbf R^T\mathbf R.
\]
The metaplectic action is linear on these quadratures.  Thus there is a real invertible matrix
$S_T$, determined by the action of $T^{-1}$ in the chosen canonical coordinates, such that
$\mathcal U_T^*\mathbf R\mathcal U_T=S_T\mathbf R$.  As quadratic forms,
\[
\mathcal U_T^*\left(N_{\mathfrak k}+\frac d2\right)\mathcal U_T
=
\frac12(S_T\mathbf R)^T(S_T\mathbf R)
\le
\|S_T\|^2\left(N_{\mathfrak k}+\frac d2\right).
\]
The last inequality follows by factoring
$\|S_T\|^2I-S_T^TS_T=C^TC$ and observing that
$\mathbf R^TC^TC\mathbf R=\sum_k(\sum_j C_{kj}R_j)^2\ge0$ as a quadratic form.
Hence a canonical vector with number expectation $R'$ is sent to one with number expectation
at most $C_T(R'+1)$.  Choosing $R'$ proportional to $R$ gives
\eqref{eq:finite-mode}.
\end{proof}

\begin{proposition}[Physical-energy stabilizers]
\label{prop:physical-stabilizers}
Let $H_1\ge0$ be self-adjoint, let $G=d\Gamma(H_1)$, and let
$z_1,z_2\in\dom H_1^{1/2}$ satisfy $\omega(z_1,z_2)=4\pi$.  There are $C,E_0<\infty$ such
that for every $E\ge E_0$ one can find a unit vector $\Psi_E\in\dom G^{1/2}$ with
\begin{equation}
\langle\Psi_E,G\Psi_E\rangle\le E,
\qquad
1-\operatorname{Re}\langle\Psi_E,-W(z_j)\Psi_E\rangle\le\frac{C}{E},
\quad j=1,2.
\label{eq:physical-stabilizers}
\end{equation}
\end{proposition}

\begin{proof}
Let $\mathfrak k=\spn_{\mathbb C}\{z_1,z_2\}$.  Since $\mathfrak k$ is finite dimensional and contained in the
form domain of $H_1$,
\[
c_{\mathfrak k}:=\sup_{\substack{\xi\in\mathfrak k\\\|\xi\|=1}}\|H_1^{1/2}\xi\|^2<\infty.
\]
On the finite-particle core over $\mathfrak k$, the quadratic form of $G$ satisfies
\begin{equation}
q_G[\Psi]\le c_{\mathfrak k}\langle\Psi,N_{\mathfrak k}\Psi\rangle.
\label{eq:G-form-number}
\end{equation}
Indeed, on each $n$-particle sector the form of $d\Gamma(H_1)$ is the sum of the one-particle
forms, each bounded by $c_{\mathfrak k}$ on $\mathfrak k$.  By closure,
\eqref{eq:G-form-number} extends to
$\dom N_{\mathfrak k}^{1/2}\subset\cF(\mathfrak k)$.  If $c_{\mathfrak k}=0$,
Lemma~\ref{lem:symplectic-transfer} already gives zero physical energy.  If
$c_{\mathfrak k}>0$, apply that lemma with number budget $E/c_{\mathfrak k}$ and absorb the
fixed constants.
\end{proof}

\subsection{The fixed-witness upper bound}

\begin{theorem}[Inverse-energy upper bound]
\label{thm:fixed-upper}
Let $H_1\ge0$, $G=d\Gamma(H_1)$, and let $u,v\in\dom H_1^{1/2}$ satisfy
$\omega(u,v)=\pi$.  For the trigonometric witness $h$ in \eqref{eq:trig-witness}, there are
$C,E_0<\infty$ such that
\begin{equation}
0\le4-\gamma_{G,E}(h)\le\frac{C}{E},
\qquad E\ge E_0.
\label{eq:fixed-upper}
\end{equation}
\end{theorem}

\begin{proof}
Set
\[
S_1=-U^2=-W(2u),
\qquad
S_2=-V^2=-W(2v).
\]
These unitaries commute because $\omega(2u,2v)=4\pi$.  Moreover
\begin{equation}
A^2=\frac{1+\operatorname{Re}S_1}{2},
\qquad
B^2=\frac{1+\operatorname{Re}S_2}{2}.
\label{eq:A2B2}
\end{equation}
Since $A^2$ and $B^2$ commute and are positive contractions,
\[
1-A^2B^2\le(1-A^2)+(1-B^2).
\]
Thus a unit vector $\Psi$ satisfying
$1-\operatorname{Re}\langle S_j\rangle_\Psi\le\eta_j$ obeys
\begin{equation}
\langle h^2\rangle_\Psi
\ge4\left(1-\frac{\eta_1+\eta_2}{2}\right)
\label{eq:h2-lower}
\end{equation}
by \eqref{eq:h-square} and \eqref{eq:A2B2}.

Conjugation by $V$ sends $h$ to $-h$ and fixes $h^2$.  Given a seed vector $\Psi$, set
\[
\rho_\Psi
=\frac12|\Psi\rangle\langle\Psi|
 +\frac12V|\Psi\rangle\langle\Psi|V^*.
\]
Then
\[
\Tr(\rho_\Psi h)=0,
\qquad
\Tr(\rho_\Psi h^2)=\langle\Psi,h^2\Psi\rangle.
\]
If the seed energy is at most $R$, Lemma~\ref{lem:weyl-energy} gives
\begin{equation}
\Tr(G\rho_\Psi)
\le R+c_v\sqrt R+\frac{c_v^2}{2},
\qquad
c_v=\frac{1}{\sqrt2}\|H_1^{1/2}v\|.
\label{eq:symmetrized-energy}
\end{equation}
Take $R=E/4$.  For $E$ sufficiently large, the right-hand side of
\eqref{eq:symmetrized-energy} is at most $E$.  Proposition~\ref{prop:physical-stabilizers},
applied to $2u,2v$ at energy $R$, gives $\eta_1+\eta_2\le C_1/E$.  Therefore
\[
\Var_{\rho_\Psi}(h)
\ge4\left(1-\frac{C_2}{E}\right),
\]
and hence
\[
\gamma_{G,E}(h)\ge4\sqrt{1-\frac{C_2}{E}}.
\]
The inequality $1-\sqrt{1-x}\le x$ for $0\le x\le1$ gives \eqref{eq:fixed-upper}.
\end{proof}

\section{Sharp one-mode quadratic asymptotics}
\label{sec:quadratic-bridge}

The general upper bound does not determine the leading coefficient.  In one canonical mode the
quadratic-energy problem can be solved to the next order.  Let
\[
R=(Q,P)^T,\qquad [Q,P]=i,\qquad
\Omega=\begin{pmatrix}0&1\\-1&0\end{pmatrix},
\]
and use the Weyl convention
\begin{equation}
W(z)=e^{-iz^T\Omega R},
\qquad
W(z)^*RW(z)=R+z.
\label{eq:canonical-weyl}
\end{equation}
For a real positive definite $2\times2$ matrix $M$, put
\begin{equation}
G_M=\frac12R^TMR-\frac12\nu,
\qquad
\nu:=\sqrt{\det M}.
\label{eq:quadratic-energy}
\end{equation}
The subtraction is the ground energy, so $G_M\ge0$.

We use the circular-variance uncertainty relation of Breitenberger
\cite{Breitenberger1985}, in the translated-unitary form obtained from Robertson's inequality
\cite{Robertson1929}.  The mixed-state proof is included because the domain point is used in
both the quadratic and second-quantized arguments.

\begin{lemma}[Breitenberger--Robertson estimate for a translated unitary]
\label{lem:translation-variance}
Let $X=X^*$ and let $S$ be unitary with $S\dom X=\dom X$ and
\begin{equation}
S^*XS=X+\ell1,
\qquad \ell\ne0.
\label{eq:shift}
\end{equation}
If $\rho$ is a normal state with $\Tr(\rho X^2)<\infty$ and
$z=\Tr(\rho S)$, then
\begin{equation}
|\ell|\,|z|
\le2\sqrt{\Var_\rho(X)}\sqrt{1-|z|^2},
\label{eq:translation-variance}
\end{equation}
and consequently
\begin{equation}
1-\operatorname{Re}z
\ge
\frac{\ell^2}{2\{\ell^2+4\Var_\rho(X)\}}.
\label{eq:translation-gap}
\end{equation}
\end{lemma}

\begin{proof}
Write
\[
C:=\frac{S+S^*}{2},
\qquad
D:=\frac{S-S^*}{2i}.
\]
Since $S\dom X=\dom X$, both $S$ and $S^*$ preserve $\dom X$, and hence so do $C$ and
$D$.  They are commuting self-adjoint contractions,
$C^2+D^2=1$, and
\[
z=\Tr(\rho C)+i\Tr(\rho D).
\]
The covariance relation \eqref{eq:shift} gives, as identities on $\dom X$,
\[
[X,C]=i\ell D,
\qquad
[X,D]=-i\ell C.
\]
To justify the mixed-state Robertson inequalities with the unbounded operator $X$, take a
spectral decomposition
$\rho=\sum_n p_n|\psi_n\rangle\langle\psi_n|$.  Since
$\Tr(\rho X^2)<\infty$, every $\psi_n$ with $p_n>0$ lies in $\dom X$ and
$\sum_n p_n\|X\psi_n\|^2<\infty$.  Thus the purification
\[
\Omega:=\sum_n\sqrt{p_n}\,\psi_n\otimes e_n
\]
belongs to $\dom(X\otimes1)$.  Applying Robertson's inequality in this purification to
$X\otimes1$ and, respectively, $C\otimes1$ and $D\otimes1$ yields
\[
\Var_\rho(X)\Var_\rho(C)
\ge\frac{\ell^2}{4}\bigl(\Tr(\rho D)\bigr)^2,
\qquad
\Var_\rho(X)\Var_\rho(D)
\ge\frac{\ell^2}{4}\bigl(\Tr(\rho C)\bigr)^2.
\]
Adding the two inequalities and using
\[
\Var_\rho(C)+\Var_\rho(D)=1-|z|^2
\]
gives
\[
\Var_\rho(X)(1-|z|^2)\ge\frac{\ell^2}{4}|z|^2,
\]
which is \eqref{eq:translation-variance}.  Equivalently,
\[
|z|^2\le\frac{4V}{\ell^2+4V},
\qquad V=\Var_\rho(X).
\]
Finally,
$1-\operatorname{Re}z\ge1-|z|\ge(1-|z|^2)/2$, which gives
\eqref{eq:translation-gap}.
\end{proof}

For a normal state with finite second moments, write
\[
m_\rho=\Tr(\rho R),\qquad
(V_\rho)_{jk}=\frac12\Tr\rho\{R_j-(m_\rho)_j,R_k-(m_\rho)_k\}.
\]
The uncertainty relation gives $V_\rho+i\Omega/2\ge0$, hence $V_\rho>0$.

\begin{lemma}[Covariance form]
\label{lem:covariance-form}
Let $\rho$ have finite second moments and covariance $V$.  Then, for every
$z\in\mathbb R^2$,
\begin{equation}
1-|\Tr(\rho W(z))|
\ge g\!\left(z^TV^{-1}z\right),
\qquad
g(t):=\frac{t}{2(t+4)}.
\label{eq:covariance-form}
\end{equation}
\end{lemma}

\begin{proof}
For $a\in\mathbb R^2$ set $X_a=a^TR$.  Equation \eqref{eq:canonical-weyl} gives
$W(z)^*X_aW(z)=X_a+a^Tz$, while
$\Var_\rho(X_a)=a^TVa$.  Lemma~\ref{lem:translation-variance} therefore yields
\[
1-|\Tr(\rho W(z))|
\ge
\frac{(a^Tz)^2}{2\{(a^Tz)^2+4a^TVa\}}.
\]
The right-hand side is increasing in the Rayleigh quotient
$(a^Tz)^2/(a^TVa)$, whose supremum over $a\ne0$ is $z^TV^{-1}z$.
\end{proof}

Let $u,v\in\mathbb R^2$ satisfy $u^T\Omega v=\pi$, and define the fixed witness by
\eqref{eq:trig-witness}.  Put $z_1=2u$, $z_2=2v$ and
\begin{equation}
F_{u,v}:=z_1z_1^T+z_2z_2^T,
\qquad
\tau_M(u,v):=
\Tr\sqrt{M^{1/2}F_{u,v}M^{1/2}}.
\label{eq:tau-bridge}
\end{equation}
Finally, set
\begin{equation}
q_M(E):=
\inf_{\rho\in\mathfrak S_{G_M,E}}
\Tr\rho(4-h^2).
\label{eq:qM}
\end{equation}

The matrix quantity in the final coefficient is fixed by a simple balance that appears in both
the converse and the construction.

\begin{lemma}[Trace balance]
\label{lem:trace-balance}
Let $H>0$ be a real $2\times2$ matrix and $c>0$.  Then
\begin{equation}
\inf_{\substack{\Sigma>0\\ \Tr(H\Sigma^{-1})\le c}}
\Tr\Sigma
=
\frac{(\Tr\sqrt H)^2}{c}.
\label{eq:trace-balance}
\end{equation}
Equality holds for
$\Sigma=(\Tr\sqrt H/c)\,H^{1/2}$.
\end{lemma}

\begin{proof}
The trace-norm Cauchy--Schwarz inequality gives
\[
\Tr\Sigma\,\Tr(H\Sigma^{-1})\ge(\Tr\sqrt H)^2.
\]
The stated lower bound follows from the constraint.  For
$\Sigma=(\Tr\sqrt H/c)H^{1/2}$ one has
$\Tr(H\Sigma^{-1})=c$, and equality holds.
\end{proof}

\begin{lemma}[Zak quasimodes for the bridge frame]
\label{lem:zak-bridge}
With $\tau=\tau_M(u,v)$, there are normalized vectors $\psi_\varepsilon$ in the quadratic-form
domain of $G_M$ such that, as $\varepsilon\downarrow0$,
\begin{align}
\langle\psi_\varepsilon,(4-h^2)\psi_\varepsilon\rangle
&=\frac{\tau}{2}\varepsilon+O(\varepsilon^2),
\label{eq:zak-bridge-deficit}\\
\langle\psi_\varepsilon,G_M\psi_\varepsilon\rangle
&=\frac{\tau}{4\varepsilon}+O(1).
\label{eq:zak-bridge-energy}
\end{align}
\end{lemma}

\begin{proof}
Write $B_0=(z_1\ z_2)$ and $L=2\sqrt\pi$.  Since
$\det B_0=z_1^T\Omega z_2=4\pi=L^2$, the matrix $S=B_0/L$ is symplectic.  Metaplectic covariance
reduces the calculation to the pair $Le_1,Le_2$ and replaces $M$ by $M'=S^TMS$.  Put
\begin{equation}
H:=L^2M'=B_0^TMB_0.
\label{eq:H-bridge}
\end{equation}
The nonzero eigenvalues of $H$ agree with those of
$M^{1/2}F_{u,v}M^{1/2}$, so
\begin{equation}
\Tr\sqrt H=\tau.
\label{eq:tau-H}
\end{equation}

Use the Zak transform with position period $L$,
\begin{equation}
(\mathcal Z_L\psi)(x,p)
=
\sqrt{\frac{L}{2\pi}}
\sum_{r\in\mathbb Z}e^{-irLp}\psi(x+rL),
\label{eq:zak-transform}
\end{equation}
on the cell $0\le x<L$, $0\le p<2\pi/L$; see
\cite{EnglertZak2006,PantaleoniZak2023}.  On smooth functions supported inside the cell,
\begin{equation}
\mathcal Z_LQ\mathcal Z_L^{-1}=x+i\partial_p,
\qquad
\mathcal Z_LP\mathcal Z_L^{-1}=-i\partial_x,
\label{eq:zak-QP}
\end{equation}
and
\begin{equation}
W(Le_1)\longmapsto e^{-iLp},
\qquad
W(Le_2)\longmapsto e^{iLx}.
\label{eq:zak-generators}
\end{equation}
Choose
\[
x_0=p_0=\frac{\pi}{L},
\]
which lies in the interior of this cell and makes both multipliers equal to $-1$.  Set
\[
\theta_1=-L(p-p_0),
\qquad
\theta_2=L(x-x_0).
\]
The canonical squared generators are then $-e^{i\theta_1}$ and $-e^{i\theta_2}$.  Since their
Weyl cocycle is trivial, the six terms in \eqref{eq:bridge-frame} have the prescribed phases
multiplied by $e^{ic^T\theta}$, where
\[
c\in\{e_1,e_1,e_2,e_2,e_1+e_2,e_1-e_2\}.
\]
In these coordinates,
\begin{equation}
R=L D_\theta+a(\theta),
\qquad
D_\theta=-i\nabla_\theta,
\qquad
a(\theta)=\begin{pmatrix}x_0+\theta_2/L\\0\end{pmatrix}.
\label{eq:zak-local-R}
\end{equation}

Let $\chi_0\in C_c^\infty(\mathbb R^2)$ be real, supported in a small neighborhood of the
origin whose closure remains inside the chosen chart, and equal to one near the origin.  Set
\[
X_*:=\frac12H^{1/2},
\qquad
f_\varepsilon(\theta)
=c_\varepsilon\chi_0(\theta)
\exp\!\left(-\frac{1}{4\varepsilon}\theta^TX_*^{-1}\theta\right),
\]
with $c_\varepsilon$ chosen so that $\|f_\varepsilon\|=1$.  These functions are smooth and
compactly supported in the chart, hence define vectors in the common form domain of $Q$ and $P$.
The cutoff and all its derivatives are separated from the concentration point, so their
contributions are exponentially small.  The Gaussian moments are therefore
\begin{align}
\langle\theta\theta^T\rangle_{f_\varepsilon}
&=\varepsilon X_*+O(e^{-c/\varepsilon}),
\label{eq:zak-theta-moment}\\
\langle D_\theta D_\theta^T\rangle_{f_\varepsilon}
&=\frac{1}{4\varepsilon}X_*^{-1}+O(e^{-c/\varepsilon}).
\label{eq:zak-D-moment}
\end{align}
The six integer vectors above have Gram sum $4I_2$.  Since their six phase deficits represent
$2(4-h^2)$, expansion at $\theta=0$ gives
\[
2\langle4-h^2\rangle_{f_\varepsilon}
=\frac12\Tr\!\left(4I_2\,\langle\theta\theta^T\rangle_{f_\varepsilon}\right)
+O(\varepsilon^2)
=2\varepsilon\Tr X_*+O(\varepsilon^2).
\]
Thus
\[
\langle4-h^2\rangle_{f_\varepsilon}
=\varepsilon\Tr X_*+O(\varepsilon^2)
=\frac{\tau}{2}\varepsilon+O(\varepsilon^2).
\]
For the energy, use \eqref{eq:zak-local-R} at the quadratic-form level.  Its leading derivative
term is
\[
\frac{1}{8\varepsilon}\Tr(HX_*^{-1}).
\]
The cross term is purely imaginary because $f_\varepsilon$ and $a(\theta)f_\varepsilon$ are real
while $D_\theta f_\varepsilon$ is purely imaginary.  As it is also the expectation of the
self-adjoint cross operator, it is real and therefore vanishes.  The multiplication term and the ground-energy
subtraction are $O(1)$.  Since $X_*^{-1}=2H^{-1/2}$,
\[
\langle G_{M'}\rangle_{f_\varepsilon}
=\frac{\tau}{4\varepsilon}+O(1).
\]
No Zak boundary term occurs: every function entering the integrations by parts is supported in a
compact subset of the chart.  Transforming back by the metaplectic unitary proves
\eqref{eq:zak-bridge-deficit}--\eqref{eq:zak-bridge-energy}.

The choice $X_*\propto H^{1/2}$ is the equality geometry of
Lemma~\ref{lem:trace-balance}.  At leading order the deficit is $\varepsilon\Tr X$ and the energy
is $(8\varepsilon)^{-1}\Tr(HX^{-1})$, so the same matrix balance that appears in the lower bound
also determines the shape of the quasimode.
\end{proof}

\begin{proposition}[Sharp phase-fixed bridge cost]
\label{prop:sharp-bridge-cost}
With the preceding notation,
\begin{equation}
q_M(E)
=\frac{\tau_M(u,v)^2}{8E}+O(E^{-2}),
\qquad E\to\infty.
\label{eq:qM-asymptotic}
\end{equation}
\end{proposition}

\begin{proof}
Since $M>0$, finite $G_M$-energy implies finite second moments of $Q$ and $P$: if
$\lambda_{\min}(M)>0$ is the smallest eigenvalue of $M$, then
$R^TMR\ge\lambda_{\min}(M)(Q^2+P^2)$ as quadratic forms.  Let $\rho\in\mathfrak S_{G_M,E}$ and let $V$ be its covariance.  Since
\[
\Tr(G_M\rho)=\frac12\Tr(MV)+\frac12m_\rho^TMm_\rho-\frac\nu2
\]
and $M>0$,
\begin{equation}
\Tr(MV)=2\Tr(G_M\rho)+\nu-m_\rho^TMm_\rho\le 2E+\nu.
\label{eq:MV-energy-bound}
\end{equation}

For the multiset $\mathcal Z_h$ in \eqref{eq:bridge-frame}, set
\[
D_{\mathcal Z_h}(\rho)
:=\sum_{z\in\mathcal Z_h}\bigl(1-|\Tr(\rho W(z))|\bigr),
\qquad
T:=\Tr(A_hV^{-1}),
\quad
A_h:=\sum_{z\in\mathcal Z_h}zz^T.
\]
The phase-fixed identity gives
\begin{equation}
2\Tr\rho(4-h^2)\ge D_{\mathcal Z_h}(\rho).
\label{eq:phase-to-modulus}
\end{equation}
By Lemma~\ref{lem:covariance-form}, with $g(t)=t/[2(t+4)]$,
\begin{equation}
D_{\mathcal Z_h}(\rho)
\ge\sum_{z\in\mathcal Z_h}g(z^TV^{-1}z)
\ge\frac{T}{2(T+4)}.
\label{eq:D-lower-T}
\end{equation}
Here the last inequality follows from $0\le t_z\le T$ and the concavity of $g$ with $g(0)=0$.
The bridge frame satisfies
\begin{equation}
A_h=4F_{u,v}.
\label{eq:Ah-bridge}
\end{equation}
Moreover,
\begin{equation}
\Tr(A_hV^{-1})\Tr(MV)
\ge
\left[\Tr\sqrt{M^{1/2}A_hM^{1/2}}\right]^2
=4\tau_M(u,v)^2,
\label{eq:matrix-CS-bridge}
\end{equation}
by the trace-norm Cauchy--Schwarz inequality.  Using \eqref{eq:MV-energy-bound}, one obtains
\[
T\ge\frac{4\tau_M(u,v)^2}{2E+\nu}.
\]
Equations \eqref{eq:phase-to-modulus} and \eqref{eq:D-lower-T} therefore give
\begin{equation}
q_M(E)
\ge
\frac{\tau_M(u,v)^2}{8E+4\nu+4\tau_M(u,v)^2}
=
\frac{\tau_M(u,v)^2}{8E}+O(E^{-2}).
\label{eq:qM-lower}
\end{equation}

For the reverse inequality, Lemma~\ref{lem:zak-bridge} gives constants $C$ and
$\varepsilon_0>0$ such that, for $0<\varepsilon<\varepsilon_0$, the quasimode has energy at most
$\tau/(4\varepsilon)+C$.  Taking
\[
\varepsilon(E)=\frac{\tau}{4(E-C)}
\]
for large $E$ and using \eqref{eq:zak-bridge-deficit} yields
\[
q_M(E)
\le
\frac{\tau^2}{8(E-C)}+O(E^{-2})
=
\frac{\tau^2}{8E}+O(E^{-2}).
\]
Together with \eqref{eq:qM-lower}, this proves \eqref{eq:qM-asymptotic}.
\end{proof}

\begin{proof}[Proof of Theorem~\ref{thm:intro-quadratic-bridge}]
Let
\[
V_M(E):=\sup_{\rho\in\mathfrak S_{G_M,E}}\Var_\rho(h).
\]
Since $\Var_\rho(h)\le\Tr(\rho h^2)$,
\begin{equation}
V_M(E)\le4-q_M(E).
\label{eq:VM-upper-q}
\end{equation}
For the reverse inequality, set $V=W(v)$ and
\[
C_v=\frac14v^TMv.
\]
Choose a state $\rho$ arbitrarily close to the infimum defining $q_M(E-C_v)$ and set
\begin{equation}
\bar\rho
=\frac12\rho+\frac14V\rho V^*+\frac14V^*\rho V.
\label{eq:quadratic-symmetrization}
\end{equation}
Conjugation by $V$ sends $h$ to $-h$ and fixes $h^2$, so
\[
\Tr(\bar\rho h)=0,
\qquad
\Tr(\bar\rho h^2)=\Tr(\rho h^2).
\]
Using $W(v)^*RW(v)=R+v$, the two linear energy shifts in
\eqref{eq:quadratic-symmetrization} cancel and
\[
\Tr(G_M\bar\rho)=\Tr(G_M\rho)+C_v\le E.
\]
Hence
\begin{equation}
V_M(E)\ge4-q_M(E-C_v).
\label{eq:VM-lower-q}
\end{equation}
Proposition~\ref{prop:sharp-bridge-cost} now gives
\[
V_M(E)=4-\frac{\tau_M(u,v)^2}{8E}+O(E^{-2}).
\]
Since $\gamma_{G_M,E}(h)=2\sqrt{V_M(E)}$,
\[
4-\gamma_{G_M,E}(h)
=\frac{\tau_M(u,v)^2}{16E}+O(E^{-2}).
\]
Thus the coefficient in \eqref{eq:intro-quadratic-bridge} is $\tau_M(u,v)^2/16$.  To obtain the explicit form, write
$C=uu^T+vv^T$, so $F_{u,v}=4C$.  For a positive $2\times2$ matrix $B$,
$(\Tr\sqrt B)^2=\Tr B+2\sqrt{\det B}$.  Therefore
\[
\frac{\tau_M(u,v)^2}{16}
=\frac14\left(\Tr(MC)+2\sqrt{\det(MC)}\right).
\]
Now $\Tr(MC)=u^TMu+v^TMv$ and
$\det C=(u^T\Omega v)^2=\pi^2$, which gives \eqref{eq:intro-kappa}.
\end{proof}

\begin{corollary}[Square coefficient]
\label{cor:square-bridge}
For $M=I_2$, $u=(\sqrt\pi,0)^T$, and $v=(0,\sqrt\pi)^T$,
\begin{equation}
4-\gamma_{N,E}(h)=\frac{\pi}{E}+O(E^{-2}).
\label{eq:square-bridge}
\end{equation}
\end{corollary}

\section{General second-quantized energies: two-sided inverse-energy rate and zero modes}

For $G=d\Gamma(H_1)$, the exact coefficient of Section~\ref{sec:quadratic-bridge}
need not persist when the one-particle spectrum is infinite dimensional or singular at zero.
Lemma~\ref{lem:translation-variance} nevertheless supplies a coercive lower bound whenever the
energy controls a quadrature conjugate to one of the squared Weyl translations.  Together with
Theorem~\ref{thm:fixed-upper}, this gives the inverse-energy rate and isolates the zero-mode
exception.

\begin{proposition}[Coercive lower bound]
\label{prop:coercive-lower}
Let $h$ be the trigonometric witness in \eqref{eq:trig-witness} and set $S=-U^2$.  Suppose
there is a self-adjoint $X$ satisfying \eqref{eq:shift}, with
$\dom G^{1/2}\subset\dom X$, and such that
\begin{equation}
\|X\Psi\|^2\le a\|G^{1/2}\Psi\|^2+b\|\Psi\|^2,
\qquad \Psi\in\dom G^{1/2},
\label{eq:coercive}
\end{equation}
for some $a,b\ge0$.  Then for every $E>\inf\sigma(G)$,
\begin{equation}
4-\gamma_{G,E}(h)
\ge
\frac{\ell^2}{2\{\ell^2+4(aE+b)\}}.
\label{eq:coercive-lower}
\end{equation}
\end{proposition}

\begin{proof}
For $\rho\in\mathfrak S_{G,E}$, the quadratic-form inequality
\eqref{eq:coercive}, applied through the spectral decomposition of $\rho$, gives
$\Var_\rho(X)\le\Tr(\rho X^2)\le aE+b$.  Lemma~\ref{lem:translation-variance} therefore
implies
\[
1-\operatorname{Re}\Tr(\rho S)\ge d_E,
\qquad
d_E:=\frac{\ell^2}{2\{\ell^2+4(aE+b)\}}.
\]
Since $h^2=4A^2B^2\le4A^2$ and
$A^2=(1+\operatorname{Re}S)/2$,
\[
\Tr(\rho h^2)\le4-2d_E.
\]
Thus
\[
\gamma_{G,E}(h)
\le4\sqrt{1-\frac{d_E}{2}}.
\]
The inequality $1-\sqrt{1-y}\ge y/2$ yields
$4-\gamma_{G,E}(h)\ge d_E$.
\end{proof}

\begin{corollary}[Two-sided inverse-energy law]
\label{cor:sharp-fixed}
Under the assumptions of Theorem~\ref{thm:fixed-upper}, suppose
\eqref{eq:coercive} holds for a quadrature satisfying \eqref{eq:shift}.  Then there are
$c,C,E_0>0$ such that
\begin{equation}
\frac{c}{E+1}\le4-\gamma_{G,E}(h)\le\frac{C}{E},
\qquad E\ge E_0.
\label{eq:sharp-fixed}
\end{equation}
\end{corollary}

\begin{lemma}[Field bound from inverse one-particle energy]
\label{lem:inverse-field-bound}
Let $H_1\ge0$ be self-adjoint on $\mathfrak h$, put $G=d\Gamma(H_1)$, and set
$\mathfrak h_+:=(\ker H_1)^\perp$.  Let $H_1^{-1/2}$ denote the inverse square root of the
restriction of $H_1$ to $\mathfrak h_+$.  If
$w\in\dom H_1^{-1/2}\subset\mathfrak h_+$, then
\[
\dom G^{1/2}\subset\dom\Phi(w)
\]
and, for every $\Psi\in\dom G^{1/2}$,
\begin{equation}
\|\Phi(w)\Psi\|^2
\le
2\|H_1^{-1/2}w\|^2\,\|G^{1/2}\Psi\|^2
+\|w\|^2\,\|\Psi\|^2.
\label{eq:inverse-field-bound}
\end{equation}
\end{lemma}

\begin{proof}
Put $\eta=H_1^{-1/2}w$, so that $\eta\in\dom H_1^{1/2}$ and
$w=H_1^{1/2}\eta$.  On the $n$-particle sector of the finite-particle form core,
Cauchy--Schwarz in the first particle variable gives
\[
 n\bigl\|(\langle w|\otimes1)\Psi_n\bigr\|^2
\le
\|\eta\|^2\,
 n\bigl\|(H_1^{1/2}\otimes1)\Psi_n\bigr\|^2.
\]
Because $\Psi_n$ is symmetric, the last factor equals the quadratic form of
$d\Gamma(H_1)$ on that sector.  Summing over $n$ yields the annihilation estimate
\[
\|a(w)\Psi\|
\le
\|H_1^{-1/2}w\|\,\|G^{1/2}\Psi\|.
\]
The canonical commutation relations also give
\[
\|a^*(w)\Psi\|^2
=
\|a(w)\Psi\|^2+\|w\|^2\|\Psi\|^2.
\]
With the normalization $\Phi(w)=(a(w)+a^*(w))/\sqrt2$, it follows that
\[
\|\Phi(w)\Psi\|^2
\le
\|a(w)\Psi\|^2+\|a^*(w)\Psi\|^2,
\]
which yields \eqref{eq:inverse-field-bound} on the finite-particle form core.  This core is dense
in $\dom G^{1/2}$ for the graph norm.  If $\Psi_n\to\Psi$ in that norm, the estimate shows that
$\Phi(w)\Psi_n$ is Cauchy.  Since $\Phi(w)$ is closed, $\Psi\in\dom\Phi(w)$, and the same
estimate follows by passage to the limit.
\end{proof}

\begin{corollary}[Zero-mode dichotomy for the fixed witness]
\label{cor:injective-energy}
Under the assumptions of Theorem~\ref{thm:fixed-upper}, exactly one of the following two
regimes occurs.
\begin{enumerate}
\item If $u\notin\ker H_1$ or $v\notin\ker H_1$, then
\begin{equation}
4-\gamma_{G,E}(h)=\Theta(E^{-1})
\qquad(E\to\infty).
\label{eq:injective-sharp}
\end{equation}
\item If $u,v\in\ker H_1$, then
\begin{equation}
\gamma_{G,E}(h)=4
\qquad\text{for every }E>0.
\label{eq:zero-mode-exact}
\end{equation}
\end{enumerate}
In particular, if $H_1$ is injective, only the first regime occurs.
\end{corollary}

\begin{proof}
Assume first that $u\notin\ker H_1$ and let $P_+$ be the orthogonal projection onto
$\mathfrak h_+=(\ker H_1)^\perp$.  Then $P_+u\ne0$, so the continuous real-linear functional
\[
w\longmapsto\omega(2u,w)
\]
is nonzero on $\mathfrak h_+$; for example it is nonzero at $w=iP_+u$.  The domain $\dom H_1^{-1/2}$ is dense in $\mathfrak h_+$: indeed,
$\mathbf 1_{[1/n,n]}(H_1)\mathfrak h\subset\dom H_1^{-1/2}$ and these spectral projections
converge strongly to $P_+$.  Hence one can choose $w\in\dom H_1^{-1/2}$ with
\[
\ell:=\omega(2u,w)\ne0.
\]
Set $X=\Phi(w)$.  The Weyl relations give, for every $t\in\mathbb R$,
\[
W(2u)^*e^{itX}W(2u)=e^{it\ell}e^{itX}.
\]
By uniqueness of the self-adjoint generator,
\[
W(2u)^*XW(2u)=X+\ell1
\]
with equality of domains.  Thus $S=-W(2u)$ preserves $\dom X$ and satisfies
\eqref{eq:shift}.  By
Lemma~\ref{lem:inverse-field-bound}, \eqref{eq:coercive} holds with
\[
a=2\|H_1^{-1/2}w\|^2,
\qquad
b=\|w\|^2.
\]
Corollary~\ref{cor:sharp-fixed} therefore applies.  If instead
$v\notin\ker H_1$, the same argument is used with $S=-V^2$ and the roles of $u$ and $v$
interchanged.

It remains to consider $u,v\in\ker H_1$.  Put
$\mathfrak k=\operatorname{span}_{\mathbb C}\{u,v\}\subset\ker H_1$.  Under the canonical
factorization
\[
\cF(\mathfrak h)\simeq \cF(\mathfrak k)\otimes\cF(\mathfrak k^\perp),
\]
the witness $h$ acts on the first factor, while
\[
d\Gamma(H_1)=1\otimes d\Gamma(H_1|_{\mathfrak k^\perp})
\]
because $H_1$ vanishes on $\mathfrak k$.  Lemma~\ref{lem:h-spectral-diameter} gives $\diam\sigma(h)=4$.  The supremum of the variance of a bounded self-adjoint operator over
normal states is one quarter of the square of its spectral diameter: the upper bound follows by
centering at the spectral midpoint, and the reverse inequality follows from equal mixtures of
vector states in arbitrarily small spectral neighborhoods of the two endpoints.  Hence the
supremum of $\Var_\rho(h)$ on normal states of $\cF(\mathfrak k)$ is $4$.  Tensoring such states with the
vacuum on $\cF(\mathfrak k^\perp)$ leaves their energy equal to zero.  Therefore
$\gamma_{G,E}(h)=4$ for every $E>0$.  The injective case follows from the first regime.
\end{proof}

\begin{proof}[Proof of Theorem~\ref{thm:intro-regular}]
The inverse-energy upper bound is Theorem~\ref{thm:fixed-upper}.  The two alternatives, including
the matching lower rate in the positive-energy case and exact saturation in the zero-mode case, are
Corollary~\ref{cor:injective-energy}.
\end{proof}

\begin{remark}[Mass-gap dependence]
\label{rem:mass-gap-dependence}
If $H_1\ge m1$ for some $m>0$, then $\ker H_1=\{0\}$ and the first regime of
Corollary~\ref{cor:injective-energy} applies.  For an explicit lower constant, choose $w$ in the finite-dimensional real
symplectic space generated by $u,v$ with $\omega(2u,w)\ne0$, set $X=\Phi(w)$, and write
$C_w:=2\|w\|^2$.  Then $X^2\le C_w(N+1)$ as quadratic forms, while
$G\ge mN$.  Proposition~\ref{prop:coercive-lower}, applied with $a=C_w/m$ and $b=C_w$, gives
\[
4-\gamma_{G,E}(h)
\ge
\frac{\ell^2}
{2\left\{\ell^2+4C_w(E/m+1)\right\}},
\qquad
\ell=\omega(2u,w)\ne0.
\]
Hence
\[
\liminf_{E\to\infty}
E\bigl(4-\gamma_{G,E}(h)\bigr)
\ge
\frac{m\ell^2}{8C_w}.
\]
Thus, for fixed witness data, the explicit lower asymptotic prefactor is linear in the mass gap.
The upper constant depends additionally on the form-energy data of the chosen Weyl pair.
\end{remark}

\begin{remark}[Connection with the dual-net setting]
The bosonic dual-net framework of \cite{Nasreddine2026DualityArxiv} identifies its local
second-order channel-order coefficient with the Weyl quantity used here.  Corollary~\ref{cor:weyl-dichotomy}
therefore gives exact finite-energy saturation on every nonzero symplectic branch.  When the
translation-covariant time-smearing construction of \cite[Thm.~5.12]{Nasreddine2026DualityArxiv}
produces directions in $\dom H_1^{1/2}$, the fixed trigonometric witnesses also fall under
Theorem~\ref{thm:fixed-upper} and Corollary~\ref{cor:injective-energy}.
\end{remark}

\section{Conclusion}

A nonzero symplectic pairing has two different energy-constrained manifestations.  If the bounded
Weyl symmetry may depend on the state, the maximal commutator variance is attained at every
admissible energy threshold.  The phase selection is essential, and in the oscillator
representation the resulting Borel symmetry can leave the number form domain.  Fixing the
trigonometric witness instead keeps the generator in the natural energy graph space and exposes a
nontrivial energy cost.

The exact identity for $4-h^2$ reduces that cost to a phase-fixed problem for a commuting Weyl
frame.  For a positive one-mode quadratic energy the corresponding all-state converse and the Zak
quasimodes meet at the same matrix balance, giving the explicit $E^{-1}$ coefficient and an
$O(E^{-2})$ remainder.  For general $d\Gamma(H_1)$ the same witness has an inverse-energy deficit
unless both directions are zero modes, in which case the endpoint is already attained at zero
energy.  The one-mode coefficient concerns the fixed commutator/stabilizer diagnostic; it is not a
code-fidelity or logical-error parameter.

\section*{Data Availability Statement}
Data sharing is not applicable to this article as no new data were created or analyzed in this study.

\bibliographystyle{unsrtnat}
\setlength{\bibsep}{2pt}
\bibliography{refs}

\end{document}